\documentclass[twoside,leqno,english,leqno]{article}
\usepackage[T1]{fontenc}
\usepackage{textcomp}
\usepackage[utf8]{inputenc}
\usepackage[letterpaper]{geometry}
\usepackage{color}
\usepackage{babel}
\usepackage{url}
\usepackage{amsbsy}
\usepackage{amstext}
\usepackage{graphicx}
\usepackage[]
 {hyperref}

\makeatletter

\@ifundefined{date}{}{\date{}}

\usepackage{siamproceedings}

\usepackage{amsfonts}
\usepackage{epstopdf}
\usepackage{enumitem}
\usepackage{algorithmic}
\ifpdf
  \DeclareGraphicsExtensions{.eps,.pdf,.png,.jpg}
\else
  \DeclareGraphicsExtensions{.eps}
\fi

\newsiamremark{remark}{Remark}
\newsiamremark{hypothesis}{Hypothesis}
\crefname{hypothesis}{Hypothesis}{Hypotheses}
\newsiamthm{claim}{Claim}

\usepackage{amsopn}

\makeatother

\begin{document}
\selectlanguage{english}

\global\long\def\relatedversion{\thanks{The full version of the paper can be accessed at \url{https://arxiv.org/abs/0000.00000}}}%
\title{Can quantum dynamics emerge from classical chaos?}
\author{Frédéric Faure\thanks{{\small Institut Fourier, UMR 5582, Laboratoire de Mathématiques, Université
Grenoble Alpes, CS 40700, 38058 Grenoble cedex 9, France , \protect\href{mailto:frederic.faure@univ-grenoble-alpes.fr}{frederic.faure@univ-grenoble-alpes.fr},}
\protect\url{https://www-fourier.ujf-grenoble.fr/~faure/}} }

\maketitle

\fancyfoot[R]{\scriptsize{Copyright ©\ 20XX by SIAM}\\
{\scriptsize{} Unauthorized reproduction of this article is prohibited}}{\scriptsize\par}




\begin{abstract}
Anosov geodesic flows are among the simplest mathematical models of
deterministic chaos. In this survey we explain how, quite unexpectedly,
quantum dynamics emerges from purely classical correlation functions.
The underlying mechanism is the discrete Pollicott--Ruelle spectrum
of the geodesic flow, revealed through microlocal analysis. This spectrum
naturally arranges into vertical bands; when the rightmost band is
separated from the rest by a gap, it governs an effective dynamics
that mirrors quantum evolution.
\end{abstract}
\global\long\def\eq#1{\underset{(#1)}{=}}%

\section{Introduction.}

\subsection{Problem of prediction.}

A fundamental question in science, central to dynamical systems theory,
is whether one can predict the long-term behavior of a system from
its short-term evolution law. Such questions arise not only in physics
but also in pure mathematics, including areas such as arithmetic and
geometry.

\subsection{Example in arithmetic.}

If the evolution law is simple to write down, one might expect the
long--time behavior to be simple to predict. In practice this is
rarely the case, and this is the main difficulty encountered in the
``theory of deterministic chaos.'' To illustrate this, let us present
a very simple example of two closely related deterministic laws, each
defining a sequence $t\in\mathbb{N}\;\longmapsto\;x(t)\in\mathbb{N}$
from an initial value $x(0)\in\mathbb{N}$.

\emph{First law.} 
\begin{equation}
x\left(t+1\right)=\begin{cases}
\frac{1}{2}x\left(t\right) & \mbox{ if }x\left(t\right)\mbox{ is even}\\
\frac{1}{2}\left(x\left(t\right)+1\right) & \mbox{ if }x\left(t\right)\mbox{ is odd}
\end{cases}\label{eq:suite_geom}
\end{equation}
This simple rule is called the \emph{geometric series}. One can predict
explicitly that for any $t\in\mathbb{N}$, $x(t)=\Bigl\lceil\tfrac{1}{2^{t}}\,x(0)\Bigr\rceil,$where
$\lceil\cdot\rceil$ denotes the ceiling function. Thus the sequence
decays monotonically, and for any $t\geq\tfrac{\ln x(0)}{\ln2}$ one
has $x(t)=1$.

\emph{Second law.} 
\begin{equation}
x\left(t+1\right)=\begin{cases}
\frac{1}{2}x\left(t\right) & \mbox{ if }x\left(t\right)\mbox{ is even}\\
\frac{1}{2}\left(3x\left(t\right)+1\right) & \mbox{ if }x\left(t\right)\mbox{ is odd}
\end{cases}\label{eq:suite_syracuse}
\end{equation}
Despite its simple appearance, it remains the famous \href{https://en.wikipedia.org/wiki/Collatz_conjecture}{Collatz conjecture}
whether, for every initial value $x(0)\in\mathbb{N}$, the sequence
eventually falls into the cycle $1,2,1,2,\ldots$. Numerically one
observes a complicated (chaotic) transient behavior. For instance,
starting from $x(0)=7$, one obtains 
\[
7,\,11,\,17,\,26,\,13,\,20,\,10,\,5,\,16,\,8,\,4,\,2,\,1,\,2,\,1,\,2,\,1,\ldots
\]
(see Figure~\ref{fig:Comportement-de-la}). 

\begin{center}
\begin{figure}
\begin{centering}
\scalebox{0.7}[0.7]{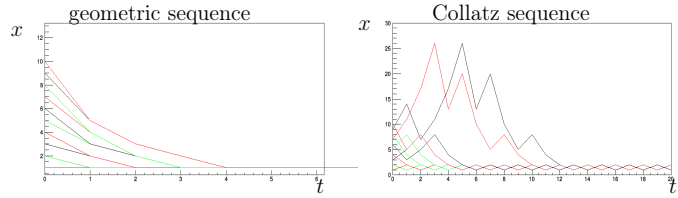}
\par\end{centering}
\caption{Illustration of the geometric sequence (\ref{eq:suite_geom}) and
the Collatz sequence (\ref{eq:suite_syracuse}), starting from different
initial values at $t=0$.}\label{fig:Comportement-de-la}
\end{figure}
\par\end{center}

\subsection{Example in geometry.}

For a smooth Riemannian manifold $(\mathcal{N},g)$, the geodesic
flow (defined precisely in Section~\ref{def:The-geodesic-vector})
takes place on the unit cotangent bundle $(T^{*}\mathcal{N})_{1}$.
In physics, the geodesic flow corresponds to the motion of a free
particle in the curved space $(\mathcal{N},g)$, i.e.\ motion without
external forces, following the ``straightest possible'' path. Intuitively,
for a smooth surface $\mathcal{N}$ embedded in $\mathbb{R}^{3}$,
starting from any point $x\in\mathcal{N}$ and any direction $V\in T_{x}\mathcal{N}$,
the geodesic flow can be visualized as a narrow adhesive tape on $\mathcal{N}$,
lying flat without folds and passing through $(x,V)$ (see Figure~\ref{fig:Geodesic-flow.-Adhesive}). 

\begin{figure}
\begin{centering}
\includegraphics[width=0.15\columnwidth]{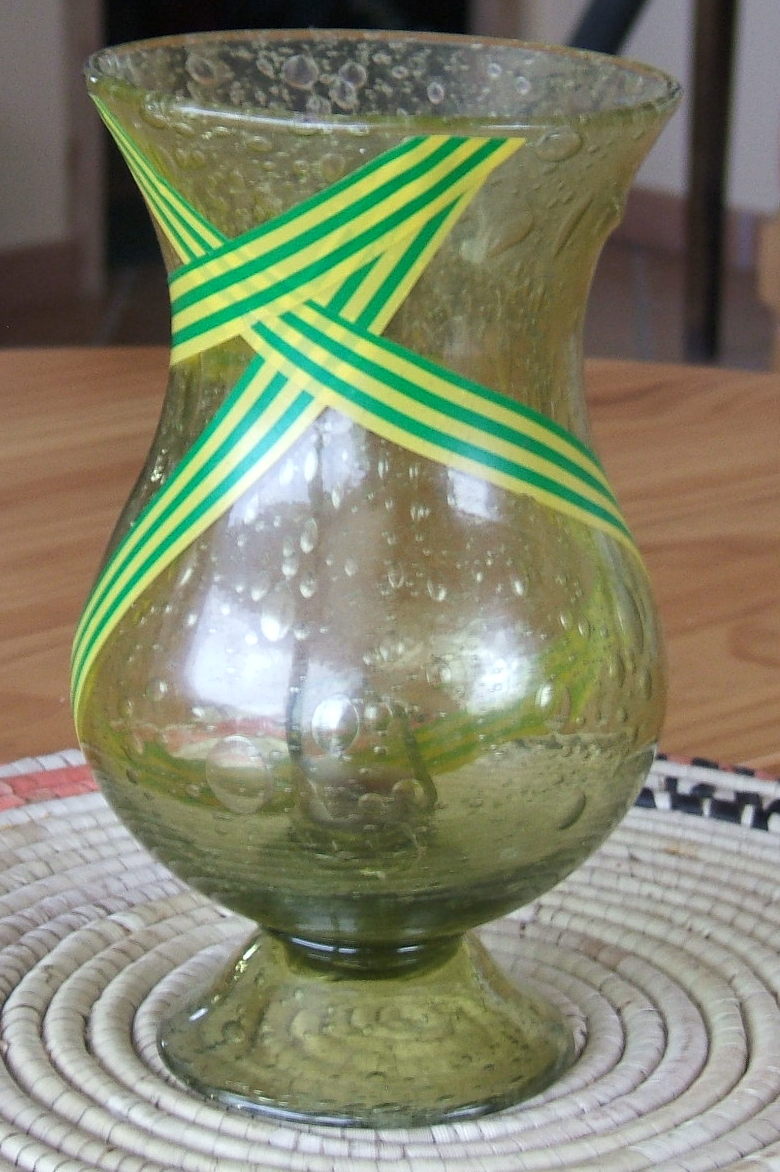}$\quad$\scalebox{0.7}[0.7]{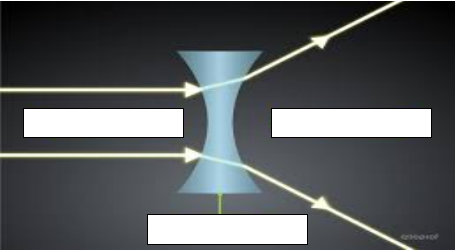}
\par\end{centering}
\caption{{\small\textbf{}}An adhesive tape laid on a vase $\mathcal{N}$ follows
the geodesics of the surface. Similarly, a light beam follows geodesics
in a medium with spatially varying velocity $c(q)>0$. In this case,
the metric is conformal to the Euclidean metric: $g_{q}^{*}\;=\;c(q)\,g_{\mathrm{Eucl}}^{*}.$}\label{fig:Geodesic-flow.-Adhesive}
\end{figure}

Despite the simplicity of this deterministic evolution law, the long--time
behavior of the geodesic flow may be difficult to predict. Consider,
for example, a closed hyperbolic surface $(\mathcal{N},g)$, that
is, one whose Gauss curvature is $\kappa=-1$ at every point. Figure~\ref{fig:Geodesic-trajectories-on}(a)
shows a single geodesic which appears highly chaotic, despite its
unique and deterministic nature. This illustrates the so--called
paradox of ``deterministic chaos'' (see also the videos~\cite{YouTubePlaylist-FredLeProf}).

Instead of considering the evolution of a single particle, one may
simultaneously follow a smooth cloud of independent particles starting
from very close initial conditions. This situation will be described
in precise mathematical terms in Section~\ref{sec:Correlation-functions}
below. As illustrated in Figure~\ref{fig:Geodesic-trajectories-on}(b),
the cloud rapidly spreads over the entire phase space $(T^{*}\mathcal{N})_{1}$,
which explains the previous impression of unpredictability for a single
trajectory. Theorem~\ref{thm:If--has} below shows that the cloud
spreads and converges exponentially towards the uniform measure on
$(T^{*}\mathcal{N})_{1}$. This phenomenon is called \emph{exponential
mixing}, and it expresses the chaotic character of the geodesic flow.

If one looks more closely, one observes fluctuations of the density
of this cloud around the equilibrium state. These fluctuations decay
exponentially like $e^{-t/2}$, where $t$ denotes the evolution time.
If one amplifies these fluctuations by the factor $e^{t/2}$, then,
as shown in Figure~\ref{fig:Geodesic-trajectories-on}(c), they behave
like waves on $\mathcal{N}$. This observation is formalized in Theorem~\ref{thm:.-Suppose-}
below (see also Section~\ref{sec:Pollicott-Ruelle-spectrum} for
the special case of hyperbolic surfaces), and it constitutes the main
theme of this paper, raised by the question in the title: \emph{``Can
quantum dynamics emerge from classical deterministic chaos?''} Indeed,
the wave equation is usually regarded as the quantization of the geodesic
flow, since $\mathrm{Op}\!\left(\|\xi\|\right)\;\approx\;\sqrt{\Delta}.$

\begin{figure}
\centering{}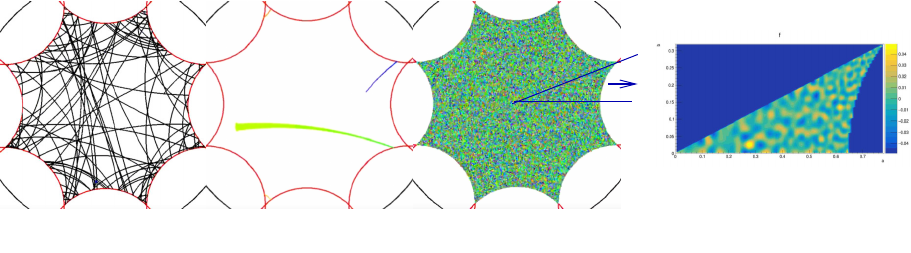\caption{Geodesic trajectories on the \protect\href{https://en.wikipedia.org/wiki/Bolza_surface}{Bolza surface},
represented as a fundamental domain in the Poincaré disk $\mathbb{D}^{2}$:
$\mathcal{N}=\Gamma\backslash\mathrm{SL}_{2}\mathbb{R}/\mathrm{SO}_{2}$,
with $\Gamma\subset\mathrm{SL}_{2}\mathbb{R}$ cocompact. Color encodes
directions. (a) A single trajectory. (b) An ensemble of nearby trajectories
spreads exponentially. (c) After some time, a small cloud of points
spreads over the whole phase space $(T^{*}\mathcal{N})_{1}$, illustrating
exponential mixing towards equilibrium. (d) Exponentially small fluctuations
around equilibrium behave like quantum waves.}\label{fig:Geodesic-trajectories-on}
\end{figure}

\subsection{Structure of the paper.}

In Section~\ref{sec:Correlation-functions}, we will introduce the
mathematical model and state several theorems that express this type
of result for specific cases such as Anosov geodesic vector fields
$X$.

The techniques involved are rich and fundamental in addressing these
questions. First, Riemannian geometry and contact geometry provide
the framework to define the geodesic flow. Next, functional analysis,
the spectral theory of non self-adjoint operators, and semigroup theory
are used to describe the evolution of smooth functions on the phase
space $(T^{*}\mathcal{N})_{1}$ over long times. But most importantly,
microlocal analysis and symplectic geometry are more than just technical
tools: they reveal in a transparent way the hidden mechanisms underlying
mixing and the fluctuations of correlations for the Anosov geodesic
flow.

A central outcome of microlocal analysis is the existence of a \emph{discrete
Pollicott--Ruelle spectrum} for the Anosov vector field $X$, which
manifests itself as resonances in the behavior of correlation functions
(see Figure~\ref{fig:Band-structure}). The exponential mixing property
is explained by the presence of a dominant isolated eigenvalue at
$z=0$ in this spectrum. The emergence of quantum dynamics, on the
other hand, is related to the existence of an internal spectral band
separated from the rest of the spectrum by a spectral gap. 

In Section~\ref{sec:Relation-with-periodic} we discuss an important
aspect of the evolution operator that is specific to deterministic
dynamical systems and expressed through the \emph{Atiyah--Bott--Guillemin
trace formula}: the formal trace of $e^{tX}$ reveals the periodic
orbits. Since the trace is independent of the choice of basis, it
can also be expressed in terms of the discrete Ruelle eigenvalues
of the vector field. From this, one can deduce further relations between
periodic orbits and the Pollicott--Ruelle spectrum, encoded in dynamical
zeta functions, which are useful for various counting problems \cite{liverani_giulietti_2012},\cite{chaubet2024closed}.

In Section~\ref{sec:Anosov-geodesic-flow} we explain the main steps
leading to the central theorems presented in this paper: the discreteness
of the Pollicott--Ruelle spectrum and its band structure. We show
how microlocal analysis is used to define anisotropic Sobolev spaces,
and how symplectic geometry enters, in particular through the linear
symplectic bundle $TT^{*}\!\left((T^{*}\mathcal{N})_{1}\right)$ with
the linearized approximation of the Anosov geodesic flow. We present
the fundamental mechanisms that appear in the linearized model, and
which can be summarized by the heuristic slogan: \emph{``quantum
mechanics is the square root of classical dynamics.''}

More technically, for a symplectic linear space $F$ (here a fiber
of $TT^{*}\!\left((T^{*}\mathcal{N})_{1}\right)$), there is a factorization
formula $L^{2}(F)\;=\;\mathrm{Spin}_{+}(F)\otimes\mathrm{Spin}_{-}(F),$
in terms of symplectic spinor spaces $\mathrm{Spin}_{\pm}(F)$ (thought
of as ``quantum spaces''), and the pushforward operator $(\Phi^{t})^{-\circ}u:=u\circ\Phi^{-t}$
of a linear symplectic map $\Phi$ on $F$ (here the linearized dynamics
of the geodesic flow in the fibers) factorizes as $(\Phi^{t})^{-\circ}\;=\;\mathrm{Op}(\Phi^{t})\otimes\bigl(\mathcal{C}\,\mathrm{Op}(\Phi^{t})\,\mathcal{C}\bigr),$
where $\mathrm{Op}(\Phi^{t})$ denotes Weyl quantization (a metaplectic
operator) and $\mathcal{C}$ is complex conjugation.

In addition, since the dynamics is hyperbolic (as is the case for
any Anosov geodesic flow), the second factor converges, for large
time, towards a rank-one projector (the ``first band''), which plays
a role analogous to the ad-hoc Szegő, Bergman, or Toeplitz projectors
in standard geometric quantization. As a consequence, the classical
dynamics $(\Phi^{t})^{-\circ}$ is effectively described by the first
factor $\mathrm{Op}(\Phi^{t})$, i.e.\ \emph{quantum dynamics emerges
dynamically.}

The final Section~\ref{sec:Informal-discussion} is devoted to an
informal discussion. 

\section{Correlation Functions for Anosov Geodesic Flows.}\label{sec:Correlation-functions}

\subsection{Anosov geodesic flow.}

We now recall the standard contact and symplectic structures underlying
the geodesic flow. Let $(\mathcal{N},g)$ be a smooth closed Riemannian
manifold of dimension $\dim\mathcal{N}=d+1\geq2$. Let 
\[
M=(T^{*}\mathcal{N})_{1}:=\{\,p\in T^{*}\mathcal{N}\;;\;\|p\|_{g}=1\,\}
\]
be the unit cotangent bundle, and let $\pi:M\to\mathcal{N}$ be the
natural projection.

We denote by $\mathcal{A}$ the canonical Liouville one-form on $M$,
defined as follows: for $p\in M$ and a tangent vector $V\in T_{p}M$,
set $q=\pi(p)$. Then $p\in T_{q}^{*}\mathcal{N}$ and $(d\pi)(V)\in T_{q}\mathcal{N}$,
and we set $\mathcal{A}_{p}(V):=p\bigl((d\pi)(V)\bigr)\in\mathbb{R}.$
The form $\mathcal{A}$ is a contact one-form, that is, $d\mathcal{A}$
is symplectic on $\ker\mathcal{A}$. The volume form $dx=(d\mathcal{A})^{d}\wedge\mathcal{A}$
is non-degenerate.

Geometrically, the contact condition means that the distribution of
hyperplanes $\ker\mathcal{A}$ is maximally non-integrable. See Figure~\ref{fig:Anosov-vector-field}. 

\begin{definition}

\label{def:The-geodesic-vector}The \emph{geodesic vector field} $X$
on $M=(T^{*}\mathcal{N})_{1}$ is defined by the conditions 
\begin{equation}
\left(d\mathcal{A}\right)\left(X,.\right)=0,\quad\mathcal{A}\left(X\right)=1,\label{eq:Reeb}
\end{equation}
that is, $X$ is the Reeb (or contact) vector field associated with
$\mathcal{A}$. Viewed as a differential operator, $X$ generates
a flow $\phi^{t}:M\to M$ at time $t\in\mathbb{R}$, defined for all
$u\in C^{\infty}(M)$ by 
\begin{equation}
\frac{d}{dt}\left(u\circ\phi^{t}\right)=X\left(u\circ\phi^{t}\right).\label{eq:def_flow}
\end{equation}

\end{definition}

\begin{remark}The one-form $\mathcal{A}$ is preserved by the geodesic
flow, since its Lie derivative satisfies 
\begin{equation}
\mathcal{L}_{X}\mathcal{A}\underset{\mathrm{Cartan}}{=}d\iota_{X}\mathcal{A}+\iota_{X}d\mathcal{A}\eq{\ref{eq:Reeb}}0.\label{eq:A_preserved}
\end{equation}
Consequently, the associated volume form is also preserved: $\mathcal{L}_{X}\,dx=0.$\end{remark}

\begin{theorem}

\label{thm:-if-}\cite{anosov_67}\cite{katok_hasselblatt}If $(\mathcal{N},g)$
has strictly negative curvature, then the geodesic vector field $X$
is \emph{Anosov}. That is, there exists a (Hölder) continuous splitting
of the tangent bundle, invariant under the differential of the flow
$d\phi^{t}$: 
\begin{equation}
TM=\mathbb{R}X\oplus E_{\mathrm{u}}\oplus E_{\mathrm{s}}.\label{eq:Anosov_Es_Eu}
\end{equation}
Moreover, for $t\gg1$, the action of $d\phi^{t}$ restricted to the
unstable direction $E_{\mathrm{u}}$ is expanding, while its action
restricted to the stable direction $E_{\mathrm{s}}$ is contracting.
In other words, the geodesic flow is \emph{uniformly hyperbolic}.

\end{theorem}

\begin{remark}\label{rem:butterfly} Theorem~\ref{thm:-if-} expresses
the property usually called \emph{sensitivity to initial conditions},
often referred to as the \href{https://en.wikipedia.org/wiki/Butterfly_effect}{Butterfly effect}.
As a consequence, each trajectory has its own ``unique story,''
when considered over the full time axis $t\in\mathbb{R}$ (past and
future).

\end{remark}

\begin{figure}
\centering{}\scalebox{0.7}[0.7]{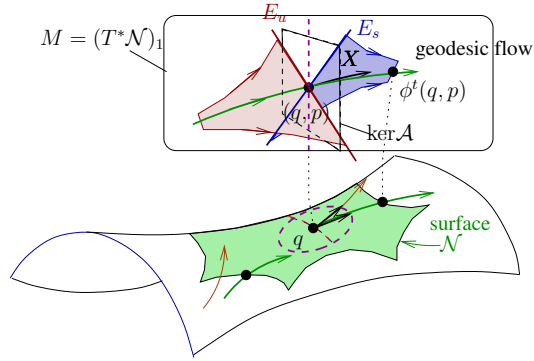}\caption{Anosov vector field $X$ on $M=(T^{*}\mathcal{N})_{1}$.}\label{fig:Anosov-vector-field}
\end{figure}

\subsection{Correlation functions.}

As observed in Figure~\ref{fig:Geodesic-trajectories-on}(a), an
individual trajectory appears unpredictable (chaotic). This is a consequence
of the instability of each trajectory, as stated in Theorem~\ref{thm:-if-}.
To overcome this difficulty, one may regard a point as a Dirac measure
and, more generally, consider the evolution of functions or distributions,
as illustrated in Figure~\ref{fig:Geodesic-trajectories-on}(c).

The action of the flow $\phi^{t}$ on functions is given by the pullback
operator. From \eqref{eq:def_flow}, this defines a continuous group
on $L^{2}(M,dx)$ with generator $X$, which we write \footnote{It is equivalent to consider the pushforward operator $(e^{tX})^{\dagger}=e^{-tX},$which
is the $L^{2}$-adjoint with respect to the invariant Liouville measure
$dx$ on $M$. The operator $e^{-tX}$ is also called the \emph{transfer
operator} (or Perron--Frobenius operator), while $e^{tX}$ is known
as the \emph{Koopman operator}. Notice that $e^{-tX}$ governs the
evolution of probability measures, in particular Dirac measures on
points: $e^{-tX}\,\delta_{x}\;=\;\delta_{\phi^{t}(x)}.$}:
\begin{equation}
e^{tX}:\begin{cases}
C^{\infty}\left(M\right) & \rightarrow C^{\infty}\left(M\right)\\
u & \mapsto u\circ\phi^{t}
\end{cases}\label{eq:def_pullback}
\end{equation}

Figures~\ref{fig:Pull-back-operator} and~\ref{fig:Geodesic-trajectories-on}(c)
illustrate the effect of the operator $e^{tX}$ on a function $u\in C^{\infty}(M)$
with small support. From these observations, one sees that for large
time $t$ the evolved function $e^{tX}u$ develops structures at very
fine scales (i.e.\ high spatial frequencies).

If we choose to disregard these fine details and observe only the
behavior that remains at the macroscopic scale, it is natural to consider
the scalar product of $e^{tX}u$ against a smooth function $v\in C^{\infty}(M)$,
viewed as an ``observable'' or a ``test function.'' The objective
is then to understand the asymptotics, as $t\to\pm\infty$, of the
function defined by 
\[
\text{(correlation function)}\qquad C_{u,v}(t):=\langle v,e^{tX}u\rangle_{L^{2}(M,dx)}.
\]

\begin{center}
\begin{figure}
\begin{centering}
\scalebox{0.7}[0.7]{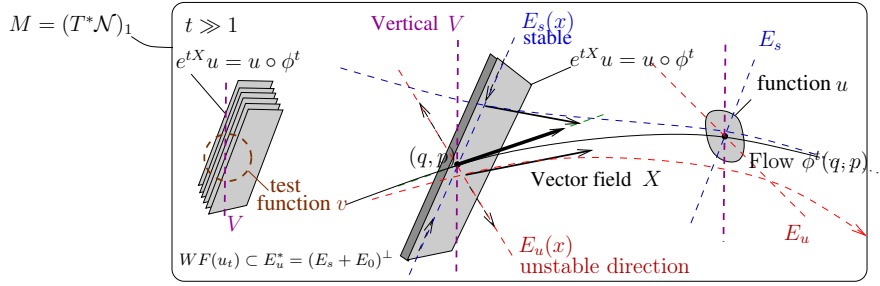}
\par\end{centering}
\caption{Action of the pullback operator $e^{tX}$ on a smooth function $u\in C^{\infty}(M)$.}\label{fig:Pull-back-operator}
\end{figure}
\par\end{center}

\subsection{Mixing.}

\begin{theorem}

\label{thm:If--has}\cite{anosov_67}\cite{liverani_contact_04}If
$(\mathcal{N},g)$ has negative curvature $\kappa<0$, then the geodesic
flow is \emph{(exponentially) mixing}. That is, for all $u,v\in C^{\infty}(M)$
one has
\begin{equation}
\langle v,e^{tX}u\rangle_{L^{2}\left(M\right)}\eq{t\gg1}\langle v,\Pi_{0}u\rangle_{L^{2}}+R\left(t\right)\label{eq:mixing}
\end{equation}
where $\Pi_{0}$ is the rank-one projector onto constant functions,
defined by 
\[
\Pi_{0}f\;=\;\boldsymbol{1}\,\Big\langle\tfrac{1}{\mathrm{Vol}(M)}\boldsymbol{1},f\Big\rangle_{L^{2}(M)}.
\]
Here the remainder satisfies $R(t)=o_{u,v}(1)$ as $t\to\infty$ (Anosov~\cite{anosov_67}).
Moreover, Liverani~\cite{liverani_contact_04} proved the sharper
estimate: there exists $\epsilon>0$ such that, for all $u,v$, $R(t)=O_{u,v}\!\left(e^{-\epsilon t}\right).$

\end{theorem}

\begin{remark}\label{rem:mixing} The mixing property~\eqref{eq:mixing}
explains the convergence towards equilibrium observed in Figures~\ref{fig:Pull-back-operator}
and~\ref{fig:Geodesic-trajectories-on}. Indeed, it shows that after
long times, the test function $v$ detects only the constant function
$\boldsymbol{1}\;\times\;\Big\langle\tfrac{1}{\mathrm{Vol}(M)}\boldsymbol{1},\,u\Big\rangle_{L^{2}(M)}.$
More precisely, it asserts that (in a weak sense) the pullback operator
converges to the rank-one projector: 
\[
e^{tX}\;\;\stackrel[t\to+\infty]{\mathrm{weak}}{\longrightarrow}\;\Pi_{0}.
\]
\end{remark}

\subsection{Pollicott-Ruelle resonances and eigenvalues.}\label{subsec:Pollicott-Ruelle-resonances}

Let us consider the line vector bundle $\mathcal{F}_{0}=\bigl|\det E_{s}\bigr|^{-1/2}\;\to\;M,$that
is, the dual of half-densities on the stable bundle $E_{s}$ (see~\eqref{eq:Anosov_Es_Eu}).
The vector field $X$ extends naturally as a differential operator
$X_{\mathcal{F}_{0}}$ on the space of sections of $\mathcal{F}_{0}$.
We then define 
\begin{equation}
\gamma_{0}^{+}:=\lim_{t\rightarrow+\infty}\log\left\Vert e^{tX_{\mathcal{F}_{0}}}\right\Vert _{L^{\infty}\left(M;\mathcal{F}_{0}\right)}^{1/t}<0.\label{eq:def_gamma_0+}
\end{equation}

For a hyperbolic surface, one has the explicit value $\gamma_{0}^{+}=-\tfrac{1}{2},$
see Section~\ref{sec:Pollicott-Ruelle-spectrum}.

Tsujii~\cite{tsujii_08} improved Theorem~\ref{thm:If--has} by
showing that, up to the decay rate $e^{\gamma_{0}^{+}t}$ (not optimal
in general), one has a finite expansion of $e^{tX}$ over finite-rank
spectral projectors. This means that there exist operators $\Pi_{j}:C^{\infty}(M)\longrightarrow\mathcal{S}'(M),\quad j\geq0,$satisfying
$\Pi_{j}\Pi_{k}=\delta_{j=k}\,\Pi_{j},$$\Pi_{j}X=X\Pi_{j}.$ Assuming,
for simplicity, that there are no Jordan blocks, this expansion reads:
$\forall u,v\in C^{\infty}\left(T_{1}^{*}\mathcal{N}\right),\forall\epsilon>0,\exists J_{\epsilon}$
such that
\begin{equation}
\langle v,e^{tX}u\rangle_{L^{2}\left(M\right)}\eq{t\gg1}\langle v,\sum_{j=0}^{J_{\epsilon}}e^{z_{j}t}\Pi_{j}u\rangle_{L^{2}}+O_{u,v}\left(e^{\left(\gamma_{0}^{+}+\epsilon\right)t}\right)\label{eq:mixing_terms}
\end{equation}
with $z_{0}=0$, which already appears in~\eqref{eq:mixing} and
corresponds to the mixing term. Possibly, only one or finitely many
projectors occur. The other resonances lie in the strip 
\[
\gamma_{0}^{+}<\Re(z_{j})<0,\qquad j\geq1,
\]
see Figure~\ref{fig:Band-structure}(a).

For hyperbolic surfaces, expansion~\eqref{eq:mixing_terms} was proved
by Ratner~\cite{ratner_87}. For the Bolza hyperbolic surface, there
are no eigenvalues with $\Re(z)>-\tfrac{1}{2}=\gamma_{0}^{+}$, except
the leading one $z_{0}=0$.

An equivalent formulation of~\eqref{eq:mixing_terms} is that the
resolvent 
\[
z\in\{\Re(z)>0\}\;\longmapsto\;(z-X)^{-1}:C^{\infty}(M)\to\mathcal{S}'(M)
\]
admits a meromorphic continuation to the half-plane $\{\Re(z)>\gamma_{0}^{+}\}$,
with poles of finite rank at $z\in\{z_{j}\}_{j}$. These poles are
called \emph{Pollicott--Ruelle resonances}, see Figure~\ref{fig:Band-structure}.

In fact, Tsujii’s result given in the next theorem is stronger: these
poles (or resonances) $(\lambda_{j})_{j}$ are discrete Pollicott--Ruelle
eigenvalues of the geodesic vector field $X$ in a suitably defined
anisotropic Sobolev space $\mathcal{H}(M)$ of distributions containing
$C^{\infty}(M)$. That is, 
\[
z\mapsto(z-X)^{-1}:\mathcal{H}(M)\rightarrow\mathcal{H}(M)
\]
is bounded except at a discrete set of poles, where the associated
spectral projectors have finite rank. 

\begin{theorem}

\label{thm:tsujii} \cite[thm 1.1]{tsujii_08}There exists an anisotropic
Hilbert space 
\[
C^{\infty}(M)\;\subset\;\mathcal{H}_{W}(M)\;\subset\;\mathcal{S}'(M),
\]
and a family of finite-rank, bounded spectral projectors $(\Pi_{j})_{j\geq0}$
such that, for every $\epsilon>0$, there exists $J_{\epsilon}$ with
\begin{equation}
\left\Vert e^{tX}-\sum_{j=0}^{J_{\epsilon}}e^{tX}\Pi_{j}\right\Vert _{\mathcal{H}_{W}\left(M\right)}\leq C_{\epsilon}e^{t\left(\gamma_{0}^{+}+\epsilon\right)}.\label{eq:expansion_tsujii}
\end{equation}

\end{theorem}

\begin{remark}\label{rem:resolvent-bounded} Expansion~\eqref{eq:expansion_tsujii}
implies that, for every $\epsilon>0$, there exists $C_{\epsilon}>0$
such that for all $z\in\mathbb{C}$ with $|\Im z|>C_{\epsilon}$ and
$\Re(z)>\gamma_{0}^{+}+\epsilon$, one has 
\[
\bigl\|(z-X)^{-1}\bigr\|_{\mathcal{H}_{W}(M)}<C_{\epsilon}.
\]
In other words, the resolvent is uniformly bounded to the right of
$\gamma_{0}^{+}$, except in neighborhoods of a finite number of eigenvalues
(see Figure~\ref{fig:Band-structure}(a)). This property is equivalent
to discreteness of the eigenvalues (or quasi-compactness) of
\[
e^{tX}:\mathcal{H}_{W}(M)\rightarrow\mathcal{H}_{W}(M)
\]
in the region $\{|z|>e^{t(\gamma_{0}^{+}+\epsilon)}\}$, for $t>0$.\end{remark}

It has been possible to extend the analysis further to the left in
the complex spectral plane: Butterley and Liverani~\cite{liverani_butterley_07}
proved that the resolvent 
\[
z\in\mathbb{C}\;\rightarrow\;(z-X)^{-1}:C^{\infty}(M)\to\mathcal{S}'(M)
\]
admits a meromorphic continuation to the whole complex plane~$\mathbb{C}$.
The poles of this continuation are the Pollicott--Ruelle resonances.
More precisely they showed that for any $C>0$, there exists an anisotropic
Sobolev space (or Banach space) $C^{\infty}(M)\subset\mathcal{H}_{W}(M)\subset\mathcal{S}'(M)$
such that 
\[
X:\mathcal{H}_{W}(M)\rightarrow\mathcal{H}_{W}(M)
\]
has a discrete Pollicott--Ruelle spectrum in the half-plane $\Re(z)>-C$.
In this framework, the Pollicott--Ruelle resonances are realized
as genuine eigenvalues of $X$ on $\mathcal{H}_{W}(M)$. These results
have been obtained after using microlocal analysis \cite{fred_flow_09}.

By contrast, there are models where one knows that the resolvent admits
a meromorphic extension (hence discrete resonances), but it is not
known whether these resonances can be realized as eigenvalues. Examples
include the Laplacian $\Delta$ on the Poincaré disk $\mathbb{D}^{2}$,
or cusps in \cite{bonthonneau2017ruelle}. 

However, these results do not provide much information on the existence
or on the precise location of the discrete Pollicott--Ruelle spectrum.
In particular, one cannot deduce an expansion such as \eqref{eq:mixing_terms}
or \eqref{eq:expansion_tsujii} for correlation functions with a remainder
term bounded by $O\left(e^{-Ct}\right),$for arbitrarily large $C>0$.

\subsection{Resonances and the Pollicott--Ruelle Discrete Spectrum in Anisotropic
Sobolev and Banach Spaces.}

Before continuing, let us give a brief and partial overview of works
and progress on resonances. A recent reference on the subject is the
book by Dyatlov and Zworski~\cite{dyatlov2019mathematical}.

As defined above in~\eqref{eq:mixing_terms}, resonances appear as
the discrete leading contributions to correlation functions. This
concept originates in physics, from the scattering of waves: for instance,
acoustic resonances in Helmholtz resonators (1863), resonances in
electromagnetism such as fluorescence, and in nuclear physics through
radioactivity with the introduction of resonant ``Gamow states''
by Gamow~ and Breit--Wigner, characterized by complex energies $E=E_{r}-i\Gamma$
and by poles of the scattering matrix.

In mathematics, starting in the 1960s, Lax and Phillips~\cite{lax1990scattering}
and others developed a rigorous framework for scattering theory and
resonances. In the early 1970s, Aguilar--Combes and Balslev--Combes
introduced the \emph{complex scaling method} \cite{aguilar_combes_71,balslev_combes_71},
which shows that on Euclidean spaces, scattering resonances can be
realized as eigenvalues in certain ``anisotropic analytic spaces.''
This method was subsequently extended in a microlocal framework by
Helffer and Sjöstrand in the 1980s~\cite{sjostrand_87}, and further
by Gérard and Martinez in the 1980s--1990s~\cite{gerard1989prolongement},
among others.

These results concern mainly the wave equation, Schrödinger-type equations,
and other PDEs in open settings or with scattering phenomena, in particular
scattering in phase space on a trapped set (see~\cite{sjostrand_87}).

In the domain of hyperbolic (Anosov) dynamical systems, the situation
is analogous, since one has scattering phenomena in the cotangent
bundle on the trapped set~\cite{fred-roy-sjostrand-07}. Ruelle~\cite{Ruelle_76,ruelle_86}
and Pollicott~\cite{pollicott1986meromorphic} introduced the concept
of Pollicott--Ruelle (PR) resonances for correlation functions.

Subsequently, many results have progressively developed the theory
of anisotropic Sobolev or Banach spaces, in which the dynamics is
modeled by an operator with discrete PR spectrum. Contributions include
the works of Rugh~\cite{rugh_92,rugh_1996}, Blank, Keller and Liverani~\cite{liverani_02},
Baladi and Tsujii~\cite{baladi_livre_00,Baladi05,baladi_05,Baladi-Tsujii08},
Butterley, Liverani and Tsujii~\cite{liverani_tsujii_06,liverani_butterley_07,tsujii_08,tsujii_FBI_10},
Gouëzel and Liverani~\cite{gouezel_2008}, Naud \cite{naud_08,naud_density_14,naud_expanding_2005},
Faure, Roy and Sjöstrand~\cite{fred-roy-sjostrand-07,fred_flow_09},
Nonnenmacher and Zworski~\cite{nonenmacher_zworski_2013}, Datchev,
Dyatlov and Zworski~\cite{dyatlov_Ruelle_resonances_2012,dyatlov_resonance_band_2013,dyatlov_zworski_zeta_2013,dyatlov2015stochastic,dyatlov2017ruelle,dyatlov2019mathematical,dyatlov_2021_pollicott_ruelle_sobolev},
Dyatlov and Guillarmou~\cite{dyatlov_guillarmou_2014}, Bonthonneau~\cite{bonthonneau2018flow},
Bonthonneau and Jézéquel for Gevrey Anosov dynamics with FBI transform~\cite{bonthonneau2020fbi},
and Jin, Tao and Zworski~\cite{jin2025counting,jin_zworski_local_trace_14},
among many others. This vast body of work has established the PR spectrum
as the natural spectral invariant of Anosov flows.

Dolgopyat~\cite{dolgopyat_98,dolgopyat_02} developed a very powerful
approach for partially hyperbolic systems (see also~\cite{liverani_contact_04}),
based on oscillatory cancellation techniques for transfer operators
which has proved extremely useful and has led to strong results, see
e.g.~\cite{tsujii_2016_exponential_mixing,tsujii2020smooth}.

Specific models have been investigated in various settings, such as
homogeneous spaces by Dyatlov, Faure, Guillarmou, Hilgert, Weich,
Wolf, and Küster~\cite{dyatlov_faure_guillarmou_2014,guillarmou_weich_resonances_16,hilgert2021higher,kuster2017quantum},
Morse--Smale flows by Dang and Rivière~\cite{dang_riviere_morse_smale_16,dang2017pollicott,dang2020spectralI,dang2020spectralII},
and dispersive billiards by Baladi, Demers and Liverani, Chaubet and
Petkov~\cite{baladi2018exponential,baladi2020thermodynamic,chaubet2022dynamical}.
There have also been important applications, including Fried's conjecture
and related topological aspects (Chaubet, Dang, Dyatlov, Guillarmou,
Rivière, Shen, Zworski, Küster and Weich)~\cite{dang2020fried,dyatlov2017ruelle,kuster2020pollicott,chaubet2024dynamical},
the marked length spectrum and geometric inverse problems, as well
as the X-ray transform and rigidity, studied by Guillarmou, Lefeuvre,
and Gouëzel~\cite{guillarmou2019marked,gouezel2021classical,lefeuvre2025microlocal},
counting closed trajectories \cite{liverani_giulietti_2012},\cite{chaubet2024closed}. 

In particular, the question of the band structure of the Ruelle spectrum
and the emergence of quantum dynamics, which is the focus of this
paper, was already investigated in~\cite{fred-PreQ-06} for a contact
extension of the cat map and in~\cite{faure-tsujii_prequantum_maps_12}
in a more general setting. In~\cite{faure-tsujii_anosov_flows_13},
the associated semiclassical zeta function is studied. More recently,
Cekic and Guillarmou~\cite{guillarmou2020band} proved the existence
of the first band for contact Anosov flows in dimension three using
horocycle operators. Below, we give details on the paper~\cite{faure-tsujii_anosov_flows_16}.

For the special case of hyperbolic manifolds endowed with an algebraic
structure, i.e. high symmetry described by Lie groups, the band spectrum
of Anosov dynamics has been studied in~\cite{dyatlov_faure_guillarmou_2014,guillarmou_weich_resonances_16,hilgert2021higher}.
In different contexts, band structure and Weyl laws for the spectrum
of resonances were investigated by Stefanov and Vodev~\cite{stefanov1995distribution},
Sjöstrand and Zworski~\cite{sjostrand1999asymptotic} for convex
obstacles, and by Dyatlov~\cite{dyatlov_resonance_band_2013} for
regular normally hyperbolic trapped sets.

\subsection{Band structure of the Pollicott Ruelle spectrum.}

In Section~\ref{subsec:Pollicott-Ruelle-resonances} we recalled
that a general Anosov vector field $X$ has a discrete spectrum of
Pollicott--Ruelle resonances, that is, discrete eigenvalues in suitable
anisotropic Sobolev spaces. At present, however, almost nothing is
known in general about the precise existence or location of these
resonances.

The situation is different for Anosov geodesic flows (or more generally
Anosov contact flows), for which the contact form $\mathcal{A}$ is
preserved (see~\eqref{eq:A_preserved}).

For $k\in\mathbb{N}$, let
\begin{equation}
\mathcal{F}_{k}=\left|\mathrm{det}E_{s}\right|^{-1/2}\otimes\mathrm{Pol}_{k}\left(E_{s}\right)\rightarrow M\label{eq:def_F_k}
\end{equation}
the finite-rank vector bundle over $M$ where $\mathrm{Pol}_{k}\left(E_{s}\right)$
denotes the space of homogeneous polynomials of degree $k$ on $E_{s}$.
We then define (as a generalization of~\eqref{eq:def_gamma_0+})
\begin{equation}
\gamma_{k}^{\pm}:=\lim_{t\rightarrow\pm\infty}\log\left\Vert e^{tX_{\mathcal{F}_{k}}}\right\Vert _{L^{\infty}\left(M;\mathcal{F}_{k}\right)}^{1/t}<0.\label{eq:deg_gamma_+-}
\end{equation}
Note that $\gamma_{k}^{+}\to-\infty$ as $k\to\infty$.

For a hyperbolic surface, one has explicitly 
\[
\gamma_{0}^{\pm}=-\tfrac{1}{2},\quad\gamma_{1}^{\pm}=-\tfrac{3}{2},\quad\gamma_{k}^{\pm}=-\tfrac{1}{2}-k.
\]
Microlocal analysis (as explained below) then leads to the following
theorem, illustrated in Figure~\ref{fig:Band-structure}. 

\begin{theorem}

\label{Thm:=000020bands}\cite[thm 1.7 p.654]{faure-tsujii_anosov_flows_16}For
any $C>0$, there exists an anisotropic Sobolev space $\mathcal{H}_{W}(M)$
such that the generator $X$ has discrete Pollicott--Ruelle spectrum
$\mathrm{Spec}(X)$ on $\{\Re z>-C\}$. Moreover, for any $\epsilon>0$
there exist constants $C_{\epsilon}>0$ and $\omega_{\epsilon}>0$
such that 
\begin{equation}
\mathrm{Spec}(X)\cap\{\Re z>-C\}\;\subset\;\Big\{|\Im z|\leq\omega_{\epsilon}\Big\}\;\cup\;\bigcup_{k\in\mathbb{N}}\Big\{\Re z\in[\gamma_{k}^{-}-\epsilon,\;\gamma_{k}^{+}+\epsilon]\Big\}.\label{eq:gaps-1}
\end{equation}
In other words, the spectrum is contained in a union of a low-frequency
horizontal band around the real axis, and vertical bands located near
the strips $[\gamma_{k}^{-},\gamma_{k}^{+}]$. Furthermore, the resolvent
is uniformly bounded in the gaps between bands:
\begin{equation}
\big\|(z-X)^{-1}\big\|_{\mathcal{H}_{W}(M)}\;\leq\;C_{\epsilon},\qquad\forall z\in\mathbb{C}\;\text{ such that }\;|\Im z|>\omega_{\epsilon},\;\Re z\in[\gamma_{k+1}^{+}+\epsilon,\;\gamma_{k}^{-}-\epsilon].\label{eq:resolvent-1}
\end{equation}

\end{theorem}

\begin{figure}
\centering{}\scalebox{0.6}[0.6]{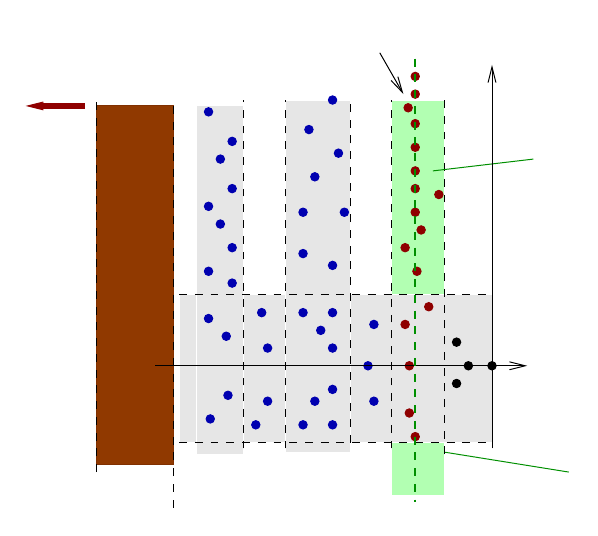}$\qquad\qquad${\small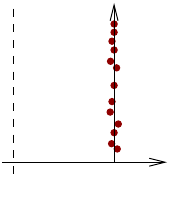}\caption{Band structure in discrete Pollicott--Ruelle spectrum}\label{fig:Band-structure}
\end{figure}

In fact, as explained in Section~\ref{sec:Anosov-geodesic-flow},
the mechanisms underlying Theorem~\ref{Thm:=000020bands} are by
now well understood: for long times an effective quantum dynamics
emerges, corresponding precisely to the restriction of $X$ to the
first band $B_{0}$.

\begin{theorem}

\label{thm:.-Suppose-}\cite[thm 1.8 p.656 and eq.(1.17)]{faure-tsujii_anosov_flows_16}.
Suppose that for some $k\in\mathbb{N}$ one has a spectral gap $\gamma_{k+1}^{+}<\gamma_{k}^{-}.$
Then for any $\epsilon>0$ there exist a constant $C_{\epsilon}>0$
and a $\mathcal{H}_{W}$-bounded spectral projector $\Pi_{[0,k]}$
of $X$ such that
\begin{equation}
\left\Vert e^{tX}-e^{tX}\Pi_{\left[0,k\right]}\right\Vert _{\mathcal{H}_{W}\left(M\right)}\leq C_{\epsilon}e^{t\left(\gamma_{k+1}^{+}+\epsilon\right)}.\label{eq:Gap}
\end{equation}

\end{theorem}

\begin{remark}\label{rem:band-gap} Under the gap assumption $\gamma_{k+1}^{+}<\gamma_{k}^{-}$,
Theorem~\ref{thm:.-Suppose-} provides a sharper description of correlation
functions (recall that $\gamma_{k}^{+}\to-\infty$ as $k\to\infty$).
In the case $k=0$, the contribution $e^{tX}\Pi_{[0,0]}$ comes entirely
from the first band $B_{0}$ and accounts for the wave dynamics (or
“emerging quantum dynamics”) observed in Figure~\ref{fig:Geodesic-trajectories-on}(d).

The projector $\Pi_{[0,k]}$ is unique up to a finite-rank correction,
corresponding to the finitely many Pollicott--Ruelle eigenvalues
lying in the interval $\Re(z)\in[\gamma_{k+1}^{+}+\epsilon,\;\gamma_{k}^{-}-\epsilon]$,
as illustrated in Figure~\ref{fig:Band-structure}.\end{remark}

\subsection{A special \textquotedblleft semiclassical\textquotedblright{} dynamical
bundle.}\label{subsec:A-special-semiclassical}

The results \eqref{eq:resolvent-1} and \eqref{eq:Gap} above rely
on the gap assumption $\gamma_{k+1}^{+}<\gamma_{k}^{-}$. There is
a relatively simple way to ensure such a gap, namely by considering
the more general problem of the group of pullback operators $\big(e^{tX_{F}}\big)_{t\in\mathbb{R}}$
acting on sections of a fiber bundle $F\to M$. This analysis was
carried out in~\cite{faure-tsujii_anosov_flows_16}, where the definition
\eqref{eq:def_F_k} is replaced by the twisted vector bundle $\mathcal{F}_{k}\;=\;\big|\det E_{s}\big|^{-1/2}\otimes\mathrm{Pol}_{k}\left(E_{s}\right)\otimes F$.
A natural choice is to take $F=|\det E_{s}|^{+1/2}$ (the bundle of
half-densities on $E_{s}$). For this choice, the first band corresponds
to the trivial bundle $\mathcal{F}_{0}=\mathbb{C}\to M$, and one
obtains 
\[
\gamma_{1}^{+}<\gamma_{0}^{\pm}=0,
\]
which gives a first band accumulating on the imaginary axis, with
a spectral gap on its left-hand side (see Figure~\ref{fig:Band-structure}(b)).

A technical difficulty arises because the bundle $E_{s}\to M$ is
not smooth. To circumvent this, in~\cite{faure-tsujii_anosov_flows_13}
we work instead with a larger but smooth Grassmannian bundle. 

\section{Pollicott-Ruelle spectrum of $X$ for a hyperbolic surface.}\label{sec:Pollicott-Ruelle-spectrum}

In this section we explain how to obtain the band structure of the
Pollicott--Ruelle spectrum (Theorem~\ref{Thm:=000020bands}) in
the special case of the geodesic flow on a compact hyperbolic surface,
using the algebraic structure of $\mathrm{SL}_{2}\mathbb{R}$. These
arguments are classical in Selberg's theory and appear prominently
in the work of Flaminio and Forni~\cite{flaminio_forni_2003}. We
follow here the approach of~\cite{dyatlov_faure_guillarmou_2014};
see also \cite{guillarmou_weich_resonances_16,hilgert2021higher,anantharaman2012intertwining,borthwick2016symmetry,Borthwick_14,Borthwick_Weich_14}.

A compact hyperbolic surface is a closed Riemannian surface $(\mathcal{N},g)$
with constant negative curvature $\kappa=-1$. An example is the Bolza
surface, illustrated in Figure~\ref{fig:Geodesic-trajectories-on}.
Such a surface can be realized as a double coset quotient 
\[
\mathcal{N}\;=\;\Gamma\backslash\mathrm{SL}_{2}\mathbb{R}/\mathrm{SO}_{2},
\]
where $\Gamma\subset\mathrm{SL}_{2}\mathbb{R}$ is a discrete cocompact
subgroup. See Figure~\ref{fig:hyp_surface}. 

\begin{figure}
\centering{}\scalebox{0.8}[0.8]{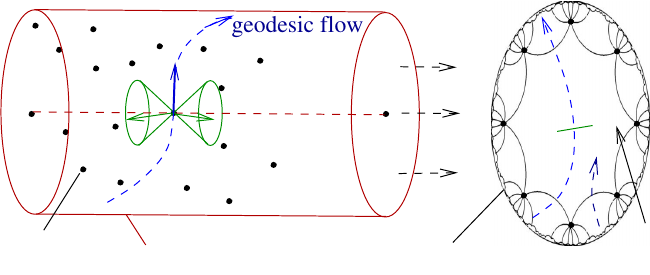}\caption{}\label{fig:hyp_surface}
\end{figure}

On $\mathrm{SL}_{2}\mathbb{R}$ we consider the three left-invariant
tangent vector fields\footnote{By definition $X$ gives a flow map: $g\mapsto ge^{tX}$ and idem
for $U,S$}
\begin{equation}
X=\left(\begin{array}{cc}
1/2 & 0\\
0 & -1/2
\end{array}\right),\quad U=\left(\begin{array}{cc}
0 & 0\\
1 & 0
\end{array}\right),\quad S=\left(\begin{array}{cc}
0 & 1\\
0 & 0
\end{array}\right)\label{eq:X,U,S}
\end{equation}
that satisfies the Lie algebra relations
\begin{equation}
\left[X,U\right]=-U,\left[X,S\right]=S,\left[S,U\right]=2X.\label{eq:Lie_algebra}
\end{equation}
A PR eigenvector $u\in\mathcal{H}_{W}\left(M\right)$ satisfies $Xu=zu$
with $\mathrm{Re}\left(z\right)\leq0$. Then
\[
X\left(Uu\right)\eq{\ref{eq:Lie_algebra}}\left(UX-U\right)u=\left(z-1\right)\left(Uu\right),\quad\left(X\right)\left(Su\right)\eq{\ref{eq:Lie_algebra}}\left(SX+S\right)u=\left(z+1\right)\left(Su\right).
\]
Hence $\exists k$ s.t. $S^{k}u=0,S^{k-1}u\neq0$. We may suppose
that $Su=0$. We use the Casimir operator $\Delta$ that commutes
with the Lie algebra:
\[
\Delta u=\left(-X^{2}-\frac{1}{2}SU-\frac{1}{2}US\right)u\eq{\ref{eq:Lie_algebra}}\left(-X^{2}-X-US\right)u=-z\left(z+1\right)u=\mu u
\]
hence $z=-\frac{1}{2}\pm i\sqrt{\mu-\frac{1}{4}}$. Consider the averaged
of $u$ under the action of $\mathrm{SO_{2}}$: we have $\left\langle u\right\rangle _{SO_{2}}\in C^{\infty}\left(\mathcal{N}\right)$
an eigenfunction of the Casimir acting on $C^{\infty}\left(\mathcal{N}\right)$
that is the hyperbolic Laplacian $\Delta\equiv-y^{2}\left(\frac{\partial^{2}}{\partial x^{2}}+\frac{\partial^{2}}{\partial y^{2}}\right)$
and has discrete spectrum $\mathrm{Spec}_{L^{2}\left(\mathcal{N}\right)}\left(\Delta\right)=\left\{ 0,\mu_{1},\mu_{2},\ldots\right\} \subset\mathbb{R}^{+}$.
We deduce that 
\begin{equation}
z=z_{0,l}=-\frac{1}{2}\pm i\sqrt{\mu_{l}-\frac{1}{4}}\label{eq:PR_hyperbolic}
\end{equation}
 and the band structure of PR spectrum, shown on figure \ref{fig:PR-discrete-spectrum}.

\begin{figure}
\begin{centering}
\scalebox{0.6}[0.6]{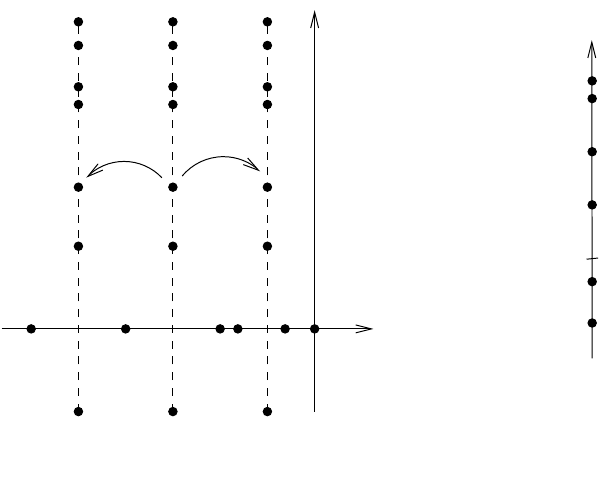}$\qquad\qquad\qquad$\scalebox{0.6}[0.6]{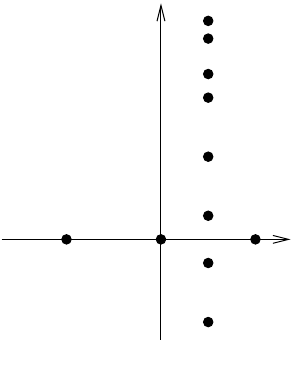}
\par\end{centering}
\caption{PR discrete spectrum of a geodesic vector field on a hyperbolic surface
$\mathcal{N}$, from the spectrum of Laplacian $\Delta$.}\label{fig:PR-discrete-spectrum}
\end{figure}

\begin{remark}For a hyperbolic surface, the bundle $E_{s}\to M$
is smooth and homogeneous. Considering the semiclassical bundle $F=\left|\det E_{s}\right|^{+1/2}$
introduced in Section~\ref{subsec:A-special-semiclassical} is equivalent
to working with the differential operator $X+\tfrac{1}{2}$. This
shifts the spectrum (see Eq.~\eqref{eq:PR_hyperbolic} and Fig.~\ref{fig:PR-discrete-spectrum})
by $+\tfrac{1}{2}$, resulting in the first band $\mathrm{B}_{0}$
lying exactly on the imaginary axis. \end{remark}

\section{Relation with periodic orbits: trace formula and dynamical zeta
functions.}\label{sec:Relation-with-periodic}

In this section we recall an exact and fundamental relation linking
the Pollicott--Ruelle spectrum to the periodic orbits of an Anosov
geodesic flow. This relation arises naturally from the fact that we
are working with the pullback operator \eqref{eq:def_pullback}, and
it plays a central role in orbit-counting problems and related applications. 

\subsection{The Atiyah-Bott-Guillemin formula.}

For $u\in C^{\infty}(M)$, $x\in M$, and $t\in\mathbb{R}$, the pullback
operator can be written as 
\[
(e^{tX}u)(x)\;\underset{\eqref{eq:def_pullback}}{=}\;u\big(\phi^{t}(x)\big)=\int_{M}K_{t}(x,x')\,u(x')\,dx',
\]
where the distributional Schwartz kernel is given by $K_{t}(x,x')\;=\;\delta\big(x'-\phi^{t}(x)\big),$
that is, the kernel is supported on the graph of the flow. Taking
the trace therefore picks out the periodic points of the flow.

A point $x\in M$ such that $\phi^{T}(x)=x$ for some $T>0$ is called
a \emph{periodic point} with period $T$. If $T>0$ is minimal, i.e.
$\phi^{t}(x)\neq x$ for all $0<t<T$, then $T$ is the \emph{primitive
period}, and the corresponding orbit $\gamma:=\{\phi^{t}(x)\;;\;t\in[0,T)\}$
is called a \emph{primitive orbit}, with length $|\gamma|:=T$. More
generally, $\phi^{n|\gamma|}(x)$ with $n\geq1$ also yields periodic
points belonging to the same orbit. 

\begin{theorem}

\label{thm:ABG-trace}\textbf{\cite{atiyah_67}\cite{guillemin_1977_lectures}}
\textbf{``Atiyah-Bott-Guillemin trace formula''}.The distributional
trace 
\[
\textrm{Tr}^{\flat}(e^{tX}):=\int_{M}K_{t}(x,x)\,dx
\]
is well defined as a distribution in $\mathcal{D}'(\mathbb{R}_{t})$,
and one has 
\begin{equation}
\mathrm{Tr}^{\flat}\left(e^{tX}\right)=\sum_{\gamma:\mathrm{primitive\,orbit}}\left|\gamma\right|\sum_{n\geq1}\frac{\delta\left(t-n\left|\gamma\right|\right)}{\left|\mathrm{det}\left(\left(\mathrm{Id}-\left(d\phi^{t}\right)\left(\gamma\right)\right)_{/E_{u}\oplus E_{s}}\right)\right|}.\label{eq:Atiyah-Boot}
\end{equation}

\end{theorem}

\begin{proof}By definition, 
\[
\mathrm{Tr}^{\flat}(e^{tX}):=\int_{M}K_{t}(x,x)\,dx=\int_{M}\delta\!\big(x-\phi^{t}(x)\big)\,dx,
\]
so the integral localizes on the periodic points of the flow. Near
a periodic point $m\in M$, we choose local coordinates $(x,z)\in\mathbb{R}^{2d}\times\mathbb{R}$
such that the hyperplane $\{z=\mathrm{const}\}$ is tangent to $E_{u}(m)\oplus E_{s}(m)$.
Then the differential reads $d\!\left((x,z)-\phi^{t}(x,z)\right)\simeq\big(\textrm{Id}-d\phi_{|E_{u}\oplus E_{s}}^{t}(x),\,0\big).$
We now use the general fact: if $f:\mathbb{R}^{n}\to\mathbb{R}^{n}$
has a fixed point $f(0)=0$, then 
\[
\int\delta(f(x))\,dx=\frac{1}{|\det Df(0)|}\int\delta(y)\,dy=\frac{1}{|\det Df(0)|}.
\]
Applying this to the flow near $m$, and integrating along the orbit,
produces the factor $|\gamma|$ corresponding to its primitive period.
This gives precisely formula~\eqref{eq:Atiyah-Boot}.\end{proof}

Just as for the trace of a matrix, the flat trace $\mathrm{Tr}^{\flat}(e^{tX})$
does not depend on the choice of a (reasonable) basis. This suggests
interpreting and using the Pollicott--Ruelle spectrum of $X$ as
providing a ``spectral decomposition'' of the operator $e^{tX}$. 

\subsection{Dynamical zeta functions.}

Dynamical zeta functions provide a convenient way to exploit the trace
formula~\eqref{eq:Atiyah-Boot} and to deduce a relation between
the Pollicott--Ruelle spectrum and periodic orbits.

\begin{theorem}

\label{thm:For-an-Anosov}\cite{liverani_giulietti_2012}\cite{dyatlov_zworski_zeta_2013}Let
$X$ be an Anosov vector field. For $\Re(z)>0$, the \emph{dynamical
zeta function} is defined by 
\begin{align}
d\left(z\right) & :=\exp\left(-\sum_{\gamma}\sum_{n\geq1}\frac{e^{-zn\left|\gamma\right|}}{n\left|\mathrm{det}\left(\left(\mathrm{Id}-\left(d\phi^{n\left|\gamma\right|}\right)\left(\gamma\right)\right)_{/E_{u}\oplus E_{s}}\right)\right|}\right)\label{eq:def_dyn_zeta}
\end{align}
where the outer sum runs over primitive periodic orbits $\gamma$
of $X$. This series converges for $\Re(z)>0$, and the function $d:\,z\mapsto d(z)$
extends holomorphically to the whole complex plane $\mathbb{C}$.
Moreover, its zeros coincide with the Pollicott--Ruelle eigenvalues,
counted with multiplicities.

\end{theorem}

\begin{remark}A similar result was obtained earlier for Anosov diffeomorphisms
by Baladi and Tsujii~\cite{Baladi-Tsujii08}.\end{remark}

Let us indicate the origin of the somewhat complicated expression
\eqref{eq:def_dyn_zeta}, and outline the proof. The guiding idea
is to mimic the case of a finite-dimensional matrix $A$ (or a finite-rank
operator), for which the eigenvalues are the zeros of the holomorphic
function $d(z)=\det(z-A)$.

Indeed, one has 
\[
(z-A)^{-1}=\int_{0}^{\infty}e^{-(z-A)t}\,dt,\qquad d(z)=\det(z-A)=\exp\!\big(\textrm{Tr}(\log(z-A))\big).
\]
Differentiating gives 
\[
\frac{d}{dz}\log d(z)=\textrm{Tr}\,(z-A)^{-1}=\int_{0}^{\infty}e^{-zt}\,\textrm{Tr}\!\big(e^{tA}\big)\,dt.
\]

For $z\notin\mathrm{Spec}(A)$, this can be integrated to yield 
\[
d(z)=d(z_{0})\cdot\exp\!\left(-\lim_{\varepsilon\to0}\int_{\varepsilon}^{\infty}\frac{1}{t}e^{-zt}\textrm{Tr}\!\big(e^{tA}\big)\,dt\;\Big|_{z_{0}}^{z}\right).
\]

Replacing $\mathrm{Tr}(e^{tA})$ by its dynamical analogue, namely
the flat trace $\mathrm{Tr}^{\flat}(e^{tX})$ from~\eqref{eq:Atiyah-Boot},
and using that $\textrm{Tr}^{\flat}(e^{tX})=0$ for $0<t<|\gamma|_{\min}$,
we obtain 
\[
\exp\!\left(-\int_{|\gamma|_{\min}}^{\infty}\frac{1}{t}e^{-zt}\,\textrm{Tr}^{\flat}(e^{tX})\,dt\right)\;\underset{\eqref{eq:Atiyah-Boot}}{=}\;\exp\!\left(-\sum_{\gamma}\sum_{n\geq1}\frac{e^{-zn|\gamma|}}{n\,\big|\det(\mathrm{Id}-(d\phi^{n|\gamma|})(\gamma))_{|E_{u}\oplus E_{s}}\big|}\right),
\]
which is exactly the definition \eqref{eq:def_dyn_zeta} of the dynamical
zeta function.

On the other hand, the poles of the resolvent $z\mapsto(z-X)^{-1}$
are precisely the Pollicott--Ruelle eigenvalues of $X$, and a careful
analysis shows that these correspond to the zeros of $d(z)$, thus
giving Theorem~\ref{thm:For-an-Anosov}. 

\subsection{Semiclassical zeta function}

As explained in section \ref{subsec:A-special-semiclassical} it is
interesting to consider the action of the geodesic vector field on
sections of the ``semi-classical'' bundle $F=\left|\mathrm{det}E_{s}\right|^{+1/2}$
(half-densities on $E_{s}$) . Indeed observe that we have $\left|\mathrm{det}\left(\left(\mathrm{Id}-\left(d\phi^{t}\right)\left(\gamma\right)\right)_{/E_{u}\oplus E_{s}}\right)\right|^{-1}\underset{t\infty}{\simeq}\mathrm{det}\left(d\phi_{/E_{u}}^{t}\right)^{-1/2}\left|\mathrm{det}\left(\left(\mathrm{Id}-\left(d\phi^{t}\right)\left(\gamma\right)\right)_{/E_{u}\oplus E_{s}}\right)\right|^{-1/2}$
and $\mathrm{det}\left(d\phi_{/E_{u}}^{t}\right)^{-1/2}=e^{-\frac{1}{2}\int^{t}\mathrm{div}X_{/E_{u}}}$
so formally 
\[
\mathrm{Tr}\left(\left(e^{tX_{F}}\right)_{/F_{\gamma}}\right)\left|\mathrm{det}\left(\left(\mathrm{Id}-\left(d\phi^{t}\right)\left(\gamma\right)\right)_{/E_{u}\oplus E_{s}}\right)\right|^{-1}\underset{t\rightarrow\infty}{\simeq}\left|\mathrm{det}\left(\left(\mathrm{Id}-\left(d\phi^{t}\right)\left(\gamma\right)\right)_{/E_{u}\oplus E_{s}}\right)\right|^{-1/2}.
\]

The ``\textbf{semi-classical zeta function}'' has been defined by
Voros (and others) \cite{voros_zeta_function_1988} by
\begin{align}
d_{sc}\left(z\right):= & \exp\left(-\sum_{\gamma}\sum_{n\geq1}\frac{e^{-zn\left|\gamma\right|}}{n\left|\mathrm{det}\left(\left(\mathrm{Id}-\left(d\phi^{t}\right)\left(\gamma\right)\right)_{/E_{u}\oplus E_{s}}\right)\right|^{1/2}}\right)\label{eq:dGV}
\end{align}

and has the following property, (see figure \ref{fig:Band-structure}
(b)). Recall that $\gamma_{1}^{+}<0$.

\begin{theorem}

\label{thm:The-semi-classical-zeta}\cite{faure-tsujii_anosov_flows_13}The
semi-classical zeta function $d_{sc}\left(z\right)$ well defined
for $\mathrm{Re}\left(z\right)\gg1$, has an meromorphic extension
on $\mathbb{C}$. For any $\forall\epsilon>0$, on $\mathrm{Re}\left(z\right)>\gamma_{1}^{+}+\epsilon$,
$d_{sc}\left(z\right)$ has finite number of poles and its zeroes
coincide (up to finite number) with the Ruelle eigenvalues of $X_{F}$
that accumulate on the imaginary axis.

\end{theorem}

One motivation for studying $d_{sc}\left(z\right)$ comes from the
Gutzwiller semiclassical trace formula in quantum chaos \cite{gutzwiller}.
In particular \textbf{in the case of a hyperbolic surface}, presented
in section \ref{sec:Pollicott-Ruelle-spectrum}, $d_{sc}\left(z\right)$
coincides (up to a shift) with\textbf{ }the \textbf{Selberg zeta function}
$\zeta_{\mathrm{Selberg}}$ : we have $\left(d\phi^{n\left|\gamma\right|}\right)_{/E_{u}\oplus E_{s}}\eq{\ref{eq:X,U,S}}\left(\begin{array}{cc}
e^{\left|\gamma\right|n} & 0\\
0 & e^{-\left|\gamma\right|n}
\end{array}\right)$, see \cite{bortwick_book_07}. This gives\footnote{Put $x=e^{-\left|\gamma\right|n}$ and use that $\left|\mathrm{det}\left(\begin{array}{cc}
1-x^{-1} & 0\\
0 & 1-x
\end{array}\right)\right|^{-1/2}=x^{1/2}\left(1-x\right)^{-1}=x^{1/2}\sum_{m\geq0}x^{m}$.}, see figure \ref{fig:PR-discrete-spectrum}, 
\begin{align*}
d_{sc}\left(z\right) & \underset{(\ref{eq:dGV})}{=}\exp\left(-\sum_{\gamma}\sum_{n\geq1}\sum_{m\geq0}\frac{1}{n}e^{-n\left|\gamma\right|\left(z+\frac{1}{2}+m\right)}\right)=\prod_{\gamma}\prod_{m\geq0}\left(1-e^{-\left(z+\frac{1}{2}+m\right)\left|\gamma\right|}\right)=:\zeta_{\mathrm{Selberg}}\left(z+\frac{1}{2}\right).
\end{align*}

\section{Microlocal mechanisms underlying the band structure.}\label{sec:Anosov-geodesic-flow}

In this section we explain how microlocal analysis leads to Theorem~\ref{Thm:=000020bands}
and Theorem~\ref{thm:.-Suppose-}, emphasizing the main steps. What
follows is a short summary of the proofs presented in \cite{faure-tsujii_anosov_flows_16}.

\subsection{Microlocal analysis on $M=\left(T^{*}\mathcal{N}\right)_{1}$.}

Recall the geodesic vector field $X$ on $M$ in definition \ref{def:The-geodesic-vector}.
We have $\mathrm{dim}M=2d+1$. On $M$, we consider local flow box
coordinates $y=\left(x,z\right)\in\mathbb{R}_{x}^{2d}\times\mathbb{R}_{z}$,
i.e. s.t. $X=\frac{\partial}{\partial z}$ and dual coordinates $\eta=\left(\xi,\omega\right)\in\mathbb{R}^{2d}\times\mathbb{R}$
on the fiber $T_{y}^{*}M$.{\small{} We introduce a}{\small\textbf{
family of wave packets $\left(y,\eta\right)\in T^{*}M\rightarrow\varphi_{\left(y,\eta\right)}\in C^{\infty}\left(M\right)$}}.
Their precise definition is given in \cite{faure-tsujii_anosov_flows_16}.
Here, with some abuse of notations that forget charts and partitions
of unity, we just notice that at high frequency $\left|\eta\right|\gg1$,
their expression is closed to a Gaussian wave packet
\begin{equation}
\varphi_{\left(y,\eta\right)}\left(y'\right)\underset{\left|\eta\right|\gg1}{\approx}a\,\exp\left(i\eta.y'-\left|\frac{x'-x}{\left\langle \eta\right\rangle ^{-1/2}}\right|^{2}-\left|\frac{z'-z}{\delta}\right|^{2}\right),\qquad\left\Vert \varphi_{\left(y,\eta\right)}\right\Vert _{L^{2}\left(M\right)}\underset{\left|\eta\right|\gg1}{\approx}1.\label{eq:wave_packets}
\end{equation}
with $\delta\ll1$.We define an operator called ``wave packet transform''
(or ``FBI transform'' or ``wavelet transform'')
\begin{equation}
\mathcal{T}:\begin{cases}
C^{\infty}\left(M\right) & \rightarrow\mathcal{S}\left(T^{*}M\right)\\
u & \rightarrow\left\{ \left(y,\eta\right)\mapsto\left(\mathcal{T}u\right)\left(y,\eta\right):=\langle\varphi_{y,\eta},u\rangle_{L^{2}\left(M\right)}\right\} 
\end{cases}\label{eq:T}
\end{equation}
that has the fundamental property called ``resolution of identity'':
\[
\mathcal{T}^{\dagger}\circ\mathcal{T}=\mathrm{Id}.
\]
With the wave packet transform isometry $\mathcal{T}$, analysis of
functions $M$ is transformed as an analysis of functions on $T^{*}M$.
We will study the lifted operator of the geodesic flow $\mathcal{T}e^{tX}\mathcal{T}^{\dagger}$.
See figure \ref{fig:resol_identity}.

\begin{figure}
\centering{}\scalebox{0.4}[0.4]{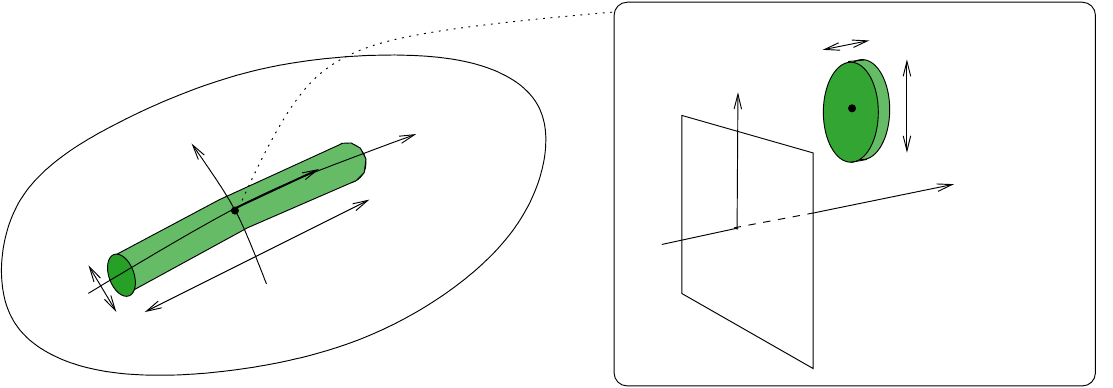}{\small ~~~~~~\scalebox{0.6}[0.6]{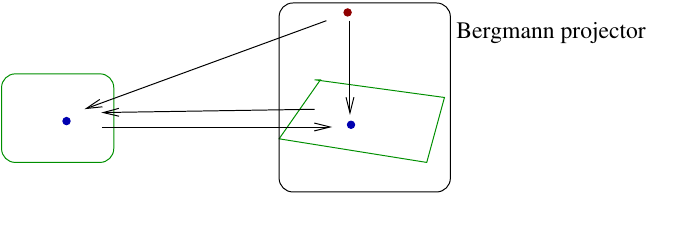}}\caption{In green : the unit ball for the metric $g$, giving the effective
size of a wave packet.}\label{fig:resol_identity}
\end{figure}

In fact the construction of wave packets (\ref{eq:wave_packets})
and the transform $\mathcal{T}$ is based on the choice of a metric\footnote{closely related to ``Hörmander metric'' in \cite{hormander1979weyl},
\cite{lerner2011metrics,nicola_rodino_livre_11}.} \textbf{$g$} on $T^{*}M$, compatible with $\Omega=dy\wedge d\eta$,{\footnotesize
\begin{equation}
g_{y,\eta}:=\left(\frac{dx}{\left\langle \eta\right\rangle ^{-1/2}}\right)^{2}+\left(\frac{d\xi}{\left\langle \eta\right\rangle ^{1/2}}\right)^{2}+\left(\frac{dz}{\delta}\right)^{2}+\left(\frac{d\omega}{\delta^{-1}}\right)^{2}.\label{eq:def_metric_g}
\end{equation}
}The geometrical meaning is that the unit boxes for the metric $g$
correspond to the effective size of wave packets and reflect the \textbf{uncertainty
principle}. Technically it is convenient to work with wave packet
transform instead of the usual Pseudo Differential Operators (PDO)
calculus \cite{taylor_tome2}, because in Anosov dynamics most of
the structures are not smooth but only Hölder continuous. With the
wave packet transform in (\ref{eq:T}), these fractal structures can
be ignored at a scale smaller that the unit ball of the metric $g$.

We have the fundamental estimate \cite{faure_tsujii_Ruelle_resonances_density_2016}
closely related to the ''\textbf{theorem of propagation of singularities}''
\begin{align}
\forall t & \in\mathbb{R},\forall N>0,\exists C_{N}>0,\forall\left(y',\eta'\right),\left(y,\eta\right)\in T^{*}M,\label{eq:prop_sing}\\
 & \left|\left\langle \delta_{\left(y',\eta'\right)},\mathcal{T}e^{tX}\mathcal{T}^{\dagger}\delta_{\left(y,\eta\right)}\right\rangle _{L^{2}\left(M\right)}\right|\leq C_{N}\left\langle \mathrm{dist}_{g}\left(\left(y',\eta'\right),\left(d\phi^{t}\right)^{*}\left(y,\eta\right)\right)\right\rangle ^{-N}.
\end{align}
Eq.(\ref{eq:prop_sing}) express that the Schwartz kernel of $\mathcal{T}e^{tX}\mathcal{T}^{\dagger}$
on $T^{*}M\times T^{*}M$ is negligible outside the graph of the induced
flow $\left(d\phi^{t}\right)^{*}:T^{*}M\rightarrow T^{*}M$.

The Anosov property (\ref{eq:Anosov_Es_Eu}) gives dual directions
$E_{u}^{*},E_{s}^{*},\Sigma=\left(E_{u}\oplus E_{s}\right)^{\perp}=\mathbb{R}\mathcal{A}$
in cotangent space and we observe that the lifted dynamics $\left(d\phi^{t}\right)^{*}$
is a \textbf{``scattering dynamics''} on the \textbf{trapped set}
$\Sigma=\mathbb{R}\mathcal{A}\subset T^{*}M$ (Liouville 1-form),
i.e. $\Sigma$ is an invariant normally hyperbolic submanifold, $\mathrm{dim}\Sigma=2\left(d+1\right)$
and $\Sigma\backslash\left\{ 0\right\} $ is symplectic since $\mathcal{A}$
is contact. See figure \ref{fig:}.

From this geometric situation we can now follow and adapt the well
established strategy of Helffer-Sjostrand 1986 \cite{sjostrand_87}
for quantum resonances (here classical PR resonances).

\subsection{Pollicott-Ruelle spectrum of $X$.}

In the outer part of $\Sigma$, where trajectories go to infinity
in the past or future, we can put a weight $W:T^{*}M\rightarrow\mathbb{R}^{+}$
such that $W\circ\left(d\phi^{t}\right)^{*}$ decays with $t\rightarrow+\infty$
and $W\equiv1$ in a vicinity of the trapped set $\Sigma$. We define
a weighed Hilbert space, called \textbf{anisotropic Sobolev space}:
\begin{equation}
\mathcal{H}_{W}\left(M\right):=\overline{\left\{ u\in C^{\infty}\left(M\right)\textnormal{\,s.t. }W\mathcal{T}u\in L^{2}\left(T^{*}M\right)\right\} }.\label{eq:def_H_W}
\end{equation}

An immediate consequence is that the operator $e^{tX}:\mathcal{H}_{W}\left(M\right)\rightarrow\mathcal{H}_{W}(M)$
has a ``negligible contribution'' outside $\Sigma$ so only the
\textbf{dynamics of $\left(d\phi^{t}\right)^{*}$ in a vicinity of
$\Sigma$ }plays a role for $e^{tX}$.

At that point, we have enough to deduce some discrete Pollicott-Ruelle
spectrum of $X:\mathcal{H}_{W}\left(M\right)\rightarrow\mathcal{H}_{W}(M)$.
This has been done as done in \cite{fred_flow_09} with usual PDO
calculus or \cite{faure_tsujii_Ruelle_resonances_density_2016} as
explained here, also in \cite{bonthonneau2020fbi} for Gevrey or analytic
Anosov flows. 

\subsection{Linearized dynamics and symplectic spinors.}

To go further and reveal the band structure of the PR spectrum, we
fix some exponent $0<\mu<1$. At every point $\rho=\omega\mathcal{A}\left(m\right)\in\Sigma$
of the trapped set, with $m\in M$ and $\omega\gg1$, we consider
a vicinity of $\rho$ of $g-$size $\asymp\omega^{\mu/2}\gg1$. From
the expression (\ref{eq:def_metric_g}) of the metric $g$, this implies
that the projection this ball on $M$ has a tiny size $\asymp\omega^{-\left(1-\mu\right)/2}\ll1$
if $\omega\gg1$. Hence this will allow us to use the \textbf{linearization
of the dynamics $d\left(d\phi^{t}\right)^{*}:TT^{*}M\rightarrow TT^{*}M$
}as a good local approximation, $\omega\gg1$. See figure \ref{fig:}.

Now we enter the setting where the relevant phenomena take place. 

\begin{figure}
\begin{centering}
\scalebox{0.7}[0.7]{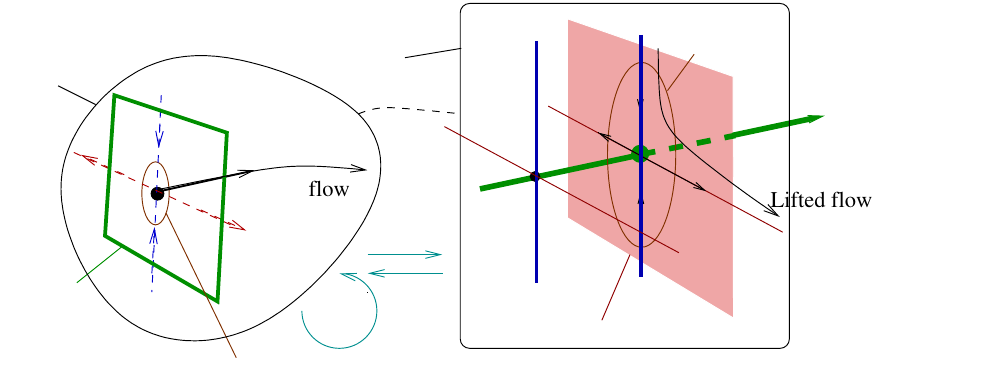}
\par\end{centering}
\caption{}\label{fig:}
\end{figure}

Recall that the trapped set $\Sigma\subset T^{*}M$ is \textbf{symplectic}
and \textbf{normally hyperbolic.} Hence at every point $\rho=\omega\mathcal{A}\left(m\right)\in\Sigma$,
there is a decomposition invariant by the dynamics $d\left(d\phi^{t}\right)^{*}$
( a ``micro-local'' decoupling)
\[
T_{\rho}T^{*}M=\underbrace{T_{\rho}\Sigma}_{\mathrm{Tangent}}\overset{\perp}{\oplus}\underbrace{\left(N_{u}\left(\rho\right)\oplus N_{s}\left(\rho\right)\right)}_{\mathrm{Normal}\,N\left(\rho\right)}
\]
and we can treat the linearized dynamics separately along the tangential
direction $T_{\rho}\Sigma$ and along the normal direction $N\left(\rho\right)$.

Here we use a simple but important property of linear symplectic maps
explained in \cite[eq.(43)]{fred-PreQ-06} and \cite[thm C.33]{faure_tsujii_Ruelle_resonances_density_2016},
which is that they factorize into two quantum (spinoral) components.
Here, concerning the linear symplectic map $\Phi:=d\left(d\phi^{t}\right)^{*}$
on $TT^{*}M$, it expresses the pull back operator $\Phi^{\circ}:u\mapsto u\circ\Phi$
as a tensor product of two metaplectic operators $\tilde{\mathrm{Op}}\left(\Phi_{T\Sigma}\right)$,
$\tilde{\mathrm{Op}}\left(\Phi_{N}\right)$, each is the \textbf{Weyl
quantization} of the linear symplectic map $\Phi$ restricted to the
respective component $T\Sigma$ or $N$:
\begin{equation}
\Phi^{\circ}=\tilde{\mathrm{Op}}\left(\Phi_{T\Sigma}\right)\otimes\tilde{\mathrm{Op}}\left(\Phi_{N}\right).\label{eq:factorization}
\end{equation}

This simple formula (\ref{eq:factorization}) has a profound meaning:
on the left, the “classical dynamics” factorizes into two quantum
dynamics. More precisely, one can show that $\Phi_{N}$ is conjugate
to $d\phi^{t}:\mathrm{Ker}(\mathcal{A})\rightarrow\mathrm{Ker}(\mathcal{A})$
and $\Phi_{T\Sigma}$ is conjugate to $\mathcal{C}(d\phi^{t})\mathcal{C}$,
where $\mathcal{C}$ denotes complex conjugation. Thus, at the level
of linear symplectic maps, quantum dynamics appears as a “square root”
of classical dynamics. For this reason, the spaces on which these
metaplectic operators act are sometimes called {*}symplectic spinors{*}
in the literature.

This factorization is, however, quite surprising: each “quantum component”
carries an uncertainty principle, but when combined their tensor product
cancels these effects, reproducing the classical pullback operator
$\Phi^{\circ}$, which transports Dirac measures to Dirac measures.
For general symplectic dynamics observed in $L^{2}$-spaces, the two
quantum components are always present, but they remain “hidden” inside
this tensor product.

\paragraph{Time evolution reveals the tangential quantum factor:}

The dynamics we are interested in is hyperbolic, and with the weight
(or escape function) $W$ in (\ref{eq:def_H_W}) we observe the two
components differently: the tangential dynamics $\tilde{\mathrm{Op}}(\Phi_{T\Sigma})$
acts in an $L^{2}$ space, while the normal dynamics $\tilde{\mathrm{Op}}(\Phi_{N})$
acts on an anisotropic Sobolev space. Recall that this anisotropic
setting is the natural framework for analyzing correlation functions
of the Anosov flow. 

Let us consider the normal operator $\tilde{\mathrm{Op}}\left(\Phi_{N}\right)$
where $\Phi_{N}$ is hyperbolic, in fact conjugated to the linear
symplectic map $d\phi^{t}:E_{u}\oplus E_{s}\rightarrow E_{u}\oplus E_{s}$.
This gives that $\tilde{\mathrm{Op}}\left(\Phi_{N}\right)$ is conjugated
to 
\begin{equation}
\mathrm{Op}^{\mathrm{Weyl}}\left(d\phi^{t}\right)=\left|\mathrm{det}\left(d\phi_{/E_{s}}^{t}\right)\right|^{1/2}\left(d\phi_{/E_{s}}^{t}\right)^{\circ}\label{eq:Op_weyl}
\end{equation}
with a contracting linear map $d\phi_{/E_{s}}^{t}$. This operator
preserves the finite rank vector bundle $\mathcal{F}_{k}=\left|\mathrm{det}E_{s}\right|^{-1/2}\otimes\mathrm{Pol}_{k}\left(E_{s}\right)$
constructed with homogeneous polynomials of degree $k$, introduced
in (\ref{eq:def_F_k}) and that belong to the anisotropic Sobolev
space. We have the estimates $\ref{eq:deg_gamma_+-}$, so that for
large time $t\gg1$, the first band $k=0$ with rank one bundle $\mathcal{F}_{0}$
emerges\footnote{A simple toy model to consider is the contracting map $d\phi^{t}:x\in\mathbb{R}\mapsto e^{-t}x\in\mathbb{R}$
for which monomials of degree $k\in\mathbb{N}$ are eigenvectors $\mathrm{Op}^{\mathrm{Weyl}}\left(d\phi^{t}\right)x^{k}\eq{\ref{eq:Op_weyl}}e^{t\left(-\frac{1}{2}-k\right)}x^{k}$
giving spectral lines for the generator at $\mathrm{Re}\left(z\right)=-\frac{1}{2}-k$
as on figure (\ref{fig:PR-discrete-spectrum}).}.

The effect of the long-time dynamics is to eliminate the second factor
in (\ref{eq:factorization}), which is asymptotically replaced by
a rank-one projector. What remains is the first, unitary quantum factor
$\tilde{\textrm{Op}}(\Phi_{T\Sigma})$ acting on the tangent space.
This is precisely how quantum dynamics emerges. 

\begin{remark}Technically, this reduction, naturally produced by
the dynamics, can be compared with the ad-hoc procedure of geometric
quantization, where one projects the evolution of functions onto a
subspace of “polarized sections” using a Szegö, Bergman, or Toeplitz
projector. The crucial difference is that in geometric quantization
the image of the projector is not invariant under the dynamics, whereas
here the projection is generated by the dynamics itself and is invariant
by definition (assuming the presence of a spectral gap). \end{remark}

\subsection{Band structure of the Pollicott-Ruelle spectrum.}

After this step, we can re-construct the global situation for the
Anosov geodesic flow that we aim to study. In \cite{faure-tsujii_anosov_flows_16},
we introduce a \textbf{semi-classical calculus} on sections $C\left(\Sigma;\mathcal{F}_{k}\right)$
with the vector bundle $\mathcal{F}_{k}=\left|\mathrm{det}E_{s}\right|^{-1/2}\otimes\mathrm{Pol}_{k}\left(E_{s}\right)$
over the symplectic trapped set $\Sigma$. We obtain associated projectors
$\mathrm{Op}_{\Sigma}\left(T_{k}\right)$ similar to the Bergman (or
Szegö ) projector in \textbf{geometric quantization,} but here, they
are (approximately) \textbf{invariant}.

Assuming $\gamma_{k+1}^{+}<\gamma_{k}^{-}$, we deduce that the discrete
Pollicott-Ruelle spectrum of $X$ in $\mathcal{H}_{W}\left(M\right)$
is contained in vertical bands $B_{k},k\in\mathbb{N}$ with gaps where
the resolvent is uniformly bounded. This gives theorem \ref{Thm:=000020bands}
and figure \ref{fig:Band-structure}.

Moreover, the \textbf{effective semi-classical calculus} on $C\left(\Sigma;\mathcal{F}_{k}\right)$
gives also a precise Weyl law (these techniques are similar to second
micro-localization techniques in other papers).

\section{Informal discussion.}\label{sec:Informal-discussion}

We conclude with an informal discussion, intended to put the results
into perspective. These considerations have motivated the work presented
in this paper. For further discussions, see \cite{fred-PreQ-06,faure-tsujii_prequantum_maps_12}. 

From a physical viewpoint, two striking similarities can be noted:
\begin{enumerate}
\item \label{enu:Theorem--shows-1} Theorem~\ref{eq:Gap} shows that the
evolution of probability measures under a deterministic yet chaotic
flow (Anosov geodesic flow) splits into an equilibrium measure plus
small fluctuations, the latter governed by an effective Schrödinger-type
equation---hence a form of “quantum dynamics” emerges.
\item \label{enu:In-physics,-experimental-1} In physics, phenomena are
described within the quantum wave formalism, with the Born rule $p(x)\,dx=|\psi(x)|^{2}\,dx,$
suggesting a probabilistic structure possibly arising from deeper
deterministic laws.
\end{enumerate}
One may then ask whether the situations in~\ref{enu:Theorem--shows-1}
and~\ref{enu:In-physics,-experimental-1} point to a deterministic
origin of quantum mechanics (see, e.g.,~\cite{nelson2012review}). 

From a mathematical viewpoint, quantization is a procedure assigning
to a symbol $a\in C^{\infty}(T^{*}\mathbb{R}_{x,\xi}^{n})$ an operator
$\mathrm{Op}(a):L^{2}(\mathbb{R}^{n})\!\to\!L^{2}(\mathbb{R}^{n})$
(see~\cite{taylor_tome2,Dumassi-Sjostrand,zworski_book_2012}). For
example, $\mathrm{Op}(\xi_{j})=-i\,\partial_{x_{j}}$. For a geodesic
flow on $(\mathcal{N},g)$, the symbol $a(x,\xi)=\|\xi\|_{g}$ yields
(up to lower-order terms) the wave operator $\sqrt{\Delta}\!\approx\!\mathrm{Op}(\|\xi\|_{g})$,
with $\Delta=d^{\dagger}d$ the Hodge Laplacian generating the wave
equation $\partial_{t}u_{t}=i\sqrt{\Delta}\,u_{t}$ (hence $\partial_{t}^{2}u_{t}=-\Delta u_{t}$).

Semiclassical analysis (WKB, Egorov, etc.) shows that in the short-wave
limit $\lambda\!\ll\!1$, solutions $u_{t}$ propagate along geodesics---
as in geometrical optics or the classical limit of quantum mechanics.
Remarkably, Theorem~\ref{eq:Gap} describes the reverse process:
\[
\text{geodesic flow}\underset{t\gg1}{\Longrightarrow}\text{quantum dynamics},
\]
whose precise meaning remains open.

Quantization is not unique: distinct pseudodifferential quantizations
sharing the same classical symbol may have different spectra (e.g.
Weyl vs. geometric quantization). In Theorem~\ref{eq:Gap}, the operator
$A=\Pi_{[0,0]}X\Pi_{[0,0]}$ acts as a canonical quantization intrinsically
defined by the Anosov flow. Its spectrum coincides with the zeros
of the dynamical zeta function (Theorems~\ref{thm:For-an-Anosov},
\ref{thm:The-semi-classical-zeta}), suggesting a privileged link
with quantum chaos (see~\cite{nonnenmacher_2011,dyatlov2022macroscopic,anantharaman2022quantum}). 

\bibliographystyle{siamplain}
\bibliography{/home/faure/articles/articles}
 
\end{document}